\documentclass[runningheads]{llncs}

\usepackage{amsmath}
\usepackage{amssymb}
\usepackage{epsf}
\usepackage{graphics}
\usepackage{epsfig}
\usepackage{latexsym}
\usepackage{enumerate}
\usepackage[normalem]{ulem}
\usepackage{color}
\usepackage{float}
\usepackage{graphicx}
\usepackage[loose]{subfigure}
\usepackage{cite}
\usepackage{algorithm}
\usepackage{algorithmic}
\usepackage{epstopdf}
\usepackage{appendix}
\usepackage{placeins}

\usepackage[pagewise]{lineno}

\newcommand{\edge}[2]{#1 #2}

\begin{document}

\title{On Optimal Beyond-Planar Graphs
\thanks{Supported by the Deutsche Forschungsgemeinschaft (DFG), grant
    Br835/20-1.}}
\titlerunning{ }
\author{Franz J. Brandenburg}
\authorrunning{Franz J. Brandenburg}
\institute{University of Passau, 94030 Passau, Germany \\
  \email{brandenb@informatik.uni-passau.de}}

\maketitle

\begin{abstract}
A graph is \emph{beyond-planar} if it can be drawn in the plane with
a specific restriction on crossings. Several types of beyond-planar
graphs have been investigated, such as $k$-planar if every edge is
crossed at most $k$ times and RAC if edges can cross only at a right
angle in a straight-line drawing. A graph is \emph{optimal}  if the
number of edges coincides with the density for its  type. Optimal
graphs are special and  are known only for some   types of
beyond-planar graphs, including 1-planar, 2-planar, and  RAC graphs.

For all types of beyond-planar graphs for which optimal graphs are
known, we compute the range for optimal graphs, establish
combinatorial properties, and show that every graph is a topological
minor of an optimal graph. Note that the minor property is
well-known  for general beyond-planar graphs.
%
\end{abstract}



\section{Introduction}
Graphs are often defined by particular properties of a drawing. The
planar graphs, in which  edge crossings are excluded, are the most
prominent example. Every $n$-vertex planar graph  has at most $3n-6$
edges. The bound is tight for triangulated planar graphs, which thus
are the \emph{optimal} planar graphs.   The planar graphs have been
characterized by the forbidden   minors $K_5$ and $K_{3,3}$
\cite{d-gt-00}. The minors cannot be a subdivision of a planar graph
\cite{K-cgdp-30}, that is a \emph{topological minor},  nor obtained
by edge contraction \cite{w-minor-37}.
  It is well-known that every
topological minor is a minor, but not conversely, see
\cite{d-gt-00}. In fact, the complete graph   $K_5$ cannot be a
topological minor of any graph of degree at most three.

There has been  recent interest in 
\emph{beyond-planar graphs} \cite{dlm-survey-beyond-19,
ht-beyond-book-20,klm-bib-17}, which are   defined by drawings with
specific restrictions on crossings. These graphs are a natural
generalization of the planar graphs.  Their study in graph theory,
graph algorithms, graph drawing, and computational geometry can
provide significant insights for the design of effective methods to
visualize real-world networks, which are non-planar, in general.

A graph is $k$-\emph{planar} \cite{pt-gdfce-97} if it has a drawing
in the plane so that each edge is   crossed by at most $k$ edges. In
particular, it is 1-planar if each edge is crossed at most
once~\cite{ringel-65}. A 1-planar drawing is   IC-\emph{planar}
(independent crossing) if every vertex is incident to at most one
crossed edge \cite{a-cnircn-08}, and NIC-\emph{planar}
(near-independent crossing) if two pairs of crossing edges share at
most one vertex \cite{z-dcmgprc-14}. A drawing  is
\emph{fan-crossing free} \cite{cpkk-fan-15} if each edge is only
crossed by independent edges, i.e., they have distinct vertices,
 and \emph{fan-crossing} \cite{b-fan-20}  if the crossing edges have a common
 vertex, i.e., they form a fan. For 1-\emph{fan-bundle graphs}
 \cite{abkks-fanbundle-18},
 edges incident to a vertex are first bundled and a  bundle can be
 crossed at most once by another bundle. An edge can only be bundled
 at one of its  vertices, that is on one side.
  These properties are topological. They hold
for \emph{embeddings}, which are equivalence classes of
topologically equivalent drawings.
 Right angle crossing  (RAC) is  a geometric property,
 in which the edges are drawn straight-line
and may cross at a right angle \cite{del-dgrac-11}. A $k$-bend RAC
drawing is a polyline drawing so that every edge is drawn with at
most $k$ bends and there is a right angle at crossings. For more
types of beyond-planar graphs, such as quasi-planar
\cite{aapps-qpg-97}, 1-gap-planar \cite{bbc-1gap-18}, fan-planar
\cite{ku-dfang-14}, and $k$-bend RAC graphs \cite{del-dgrac-11}, we
refer to \cite{dlm-survey-beyond-19}. Particular beyond-planar
graphs can be defined by first order logic formulas \cite{b-FOL-18}
and in terms of an avoidance of (natural and radial) grids
\cite{afps-grids-14}.

Beyond-planar graphs  have been studied with different intensity and
depth.   In particular, the \emph{density}, which is an upper bound
on the number of edges of $n$-vertex graphs, the size of the largest
complete (bipartite) graph \cite{abks-beyond-Kn-19}, and inclusion
relations have been investigated \cite{dlm-survey-beyond-19,
ht-beyond-book-20}. The linear density is a typical property of
beyond-planar graphs, see Table~\ref{table}. Small complete graphs
$K_k$ with $k \leq 11$ distinguish some types, see
\cite{abks-beyond-Kn-19, b-fcf-21}.
 Inclusions are canonical, in general, so
that a   restriction on drawings implies a proper inclusion for the
graph classes \cite{dlm-survey-beyond-19}. In particular, every RAC
drawing is both fan-crossing free and quasi-planar. There are mutual
incomparabilities, e.g., between RAC, fan-crossing
    and 2-planar  graphs
\cite{b-fcf-21}. Also RAC graphs are incomparable with each of
1-planar graphs \cite{el-racg1p-13}, NIC-planar graphs
\cite{bbhnr-NIC-17} and $k$-planar graphs for every fixed $k$
\cite{b-fcf-21}. However,
 every IC-planar graph is both 1-planar and RAC
\cite{bdeklm-IC-16}.

The situation is simpler for optimal graphs. A graph 
is \emph{optimal}    if its number of edges meets the established
bound on the density of graphs of its type. Hence, the density is
tight for values $n$ for which there are optimal $n$-vertex graphs.
The term optimal has been introduced by Bodendiek et
al.~\cite{bsw-bs-83} who have studied optimal 1-planar graphs
\cite{bsw-1og-84}. At other places \emph{extreme} or \emph{maximally
dense} is used.

Optimal graphs are on top of an augmentation  of graphs by
additional edges, so that the defining property of the graphs is not
violated, that is the type is preserved. Augmented graphs can often
be handled more easily. A drawing
  of a graph is (\emph{planar-maximal}) \emph{maximal}  if no further
(uncrossed) edge can be added without violation \cite{abk-sld3c-13}.
A graph $G$ is \emph{maximal} for some type $\tau$ if $G+e$ is not a
$\tau$-graph for any edge $e$ that is added to $G$. Note that there
are \emph{densest} and \emph{sparsest} graphs, which are maximal
graphs with the maximum and minimum number of edges among all
$n$-vertex graphs in their type \cite{bbhnr-NIC-17,
begghr-odm1p-13}.

We are aware of optimal graphs  for the following types  of
beyond-planar graphs: 1-planar, 2-planar, IC-planar, NIC-planar,
1-fan-bundle and RAC graphs, see Proposition~\ref{prop1} and
Table~\ref{table}. There are more types of beyond-planar graphs with
a density of $4n-8$ and $5n-10$, respectively, namely fan-crossing
free \cite{cpkk-fan-15} and grid crossing graphs \cite{b-FOL-18} as
well as  fan-planar \cite{ku-dfang-14}, fan-crossing
\cite{b-fan-20}, 1-gap planar
  \cite{bbc-1gap-18} and 5-map graphs
\cite{b-5maps-19}, where 5-map graphs are simultaneously 2-planar
and fan-crossing (see also \cite{dlm-survey-beyond-19} for
definitions), so that there are optimal graphs for these types, too.

Optimal graphs  are yet unknown for types of beyond-planar graphs
with a density above $5n-10$. 
The known bounds of $5.5 n-11$ for 3-planar   \cite{bkr-optimal-17}
and $6.5n -20$ for quasi-planar graphs \cite{at-mneqpg-07} are tight
up to a constant. Similarly, there is a constant gap for graphs with
geometric thickness two (doubly linear) \cite{hsv-rstg-99}, and for
bar-visibility and rectangle visibility graphs \cite{hsv-rstg-99}.
There are larger gaps for $k$-planar graphs with $k \geq 4$
\cite{a-cn-19, pt-gdfce-97}, 2- and 3-bend RAC graphs
\cite{del-dgrac-11}, and graphs avoiding special grids
\cite{afps-grids-14, ppst-tgnlg-05}.

 In general, the recognition of beyond-planar graphs is
NP-complete  \cite{dlm-survey-beyond-19}. However, optimal 1-planar
\cite{b-ro1plt-18}, optimal 2-planar \cite{fkr-o2p-21}, and optimal
NIC-planar graphs \cite{bbhnr-NIC-17} can be recognized in linear
time and optimal IC-planar graphs in cubic time \cite{b-IC+NIC-18}.
The recognition problem for optimal RAC graphs and optimal
1-fan-bundle graphs is open. Since an optimal RAC graph is 1-planar
and triangulated, the pairs of crossing edges can be computed in
cubic time \cite{b-4mapGraphs-19}, so that it remains to determine
whether or not all pairs of crossing edges can cross at a right
angle in a straight-line drawing.\\


\textbf{Our contribution:} In this paper, we consider optimal graph
of the aforementioned types $\tau$ of beyond-planar graphs. We study
combinatorial properties and compute the range for optimal graphs.
We show that every graph has a subdivision that is an optimal
$\tau$-graph, whereas there are optimal $\tau$-graphs that contain
$K_{42}$ as a minor but not as a topological minor.

The paper is organized as follows:  We introduce basic notions on
beyond-planar graphs in the next section and recall some properties
of such graphs. We study combinatorial properties of optimal graphs
in Section~\ref{sect:comb} and minors in  Section~\ref{sect:minor},
and we conclude in Section~\ref{sect:conclusion}.

\section{Preliminaries} \label{sect:prelim}

 We consider graphs
 that are \emph{simple} both in a graph theoretic
and in a topological sense. Thus there are no multi-edges or loops,
adjacent edges do not cross, and two edges cross at most once in a
drawing. A graph $G=(V,E)$ consists of a set of $n$ vertices and $m$
edges. We assume that it is defined by a drawing $\Gamma(G)$ in the
plane. The \emph{planar skeleton}  of $G$ (or $\Gamma(G)$) is the
subgraph induced by the uncrossed edges.

A crossed quadrangle,   X-\emph{quadrangle} for short, is a planar
quadrangle with a pair of crossing edges  in its interior. There are
no other vertices or edges in the interior, as opposed to
\cite{begghr-odm1p-13}. At several other places, the term kite has
been used. Similarly, an X-\emph{pentagon} consists of a pentagon of
five uncrossed edges and only a pentagram of five crossing edges in
its interior.  These drawings of $K_4$ and $K_5$ have been used for
optimal 1- and 2-planar graphs \cite{bkr-optimal-17, bsw-1og-84,
pt-gdfce-97}. We say that a vertex (edge) \emph{is in a triangle} 
 if it is in the boundary, and  in an X-quadrangle if is a
part of the X-quadrangle.

We consider beyond-planar graphs of \emph{type} $\tau$, where $\tau$
ranges over   the set of beyond-planar graphs for which optimal
graphs are known: IC-planar, NIC-planar, 1-planar, 2-planar,
1-fan-bundle and right angle crossing (RAC) graphs, as well as
fan-crossing free, grid crossing,
 fan-crossing,   fan-planar, 1-gap-planar  and 5-map
graphs.

For convenience, we do not distinguish between a graph, a drawing or
an embedding, and we assume that a graph and its
drawing or embedding are always of the same type.\\

Next we summarize related work on beyond-planar graphs of type
$\tau$.

\begin{proposition} \label{prop1}

\begin{enumerate} [(i)]
\item An $n$-vertex graph is optimal 1-planar if it has $4n-8$ edges
\cite{bsw-1og-84}.
 A graph $G$ is optimal 1-planar  if and only if the planar
skeleton is a 3-connected quadrangulation \cite{s-o1pgts-10}, so
that $G$ is obtained by inserting a pair of crossing edges in each
quadrangle. There are optimal 1-planar graphs if and only if $n=8$
and $n \geq 10$ \cite{bsw-1og-84}. The number of optimal 1-planar
graphs is known for $n \leq 36$ \cite{bggmtw-gsqs-05}. Optimal
1-planar graphs have a unique embedding, except for extended wheel
graphs $XW_{2k}, k \geq 3$, which have two embeddings for $k \geq 4$
\cite{s-s1pg-86}
  and eight for $k=3$ \cite{s-rm1pg-10}. The embeddings are unique up to graph
isomorphism \cite{s-rm1pg-10}. Optimal 1-planar graphs can be
recognized in linear time \cite{b-ro1plt-18}.
\item Every
optimal fan-crossing free graph is  1-planar \cite{cpkk-fan-15}, and
thus optimal 1-planar.
\item An $n$-vertex graph is optimal IC-planar if it has $\frac{13}{4}n-6$ edges
\cite{zl-spgic-13}. There are optimal IC-planar  graphs if and only
if $n=4k$ and $k \geq 2$ \cite{zl-spgic-13}.
 Optimal IC-planar graphs can be recognized in cubic
time \cite{b-IC+NIC-18}.
\item An $n$-vertex graph is optimal NIC-planar if it has $\frac{18}{5}(n-2)$ edges
\cite{z-dcmgprc-14}. There are optimal NIC-planar graphs if and only
if n=5k+2 for $k \geq 2$ \cite{bbhnr-NIC-17,cs-tc1pg-14}. Optimal
NIC-planar graphs have a unique NIC-planar embedding and
  can be recognized in
linear time \cite{bbhnr-NIC-17}.
\item An $n$-vertex graph is optimal 2-planar if it has $5n-10$
edges. A graph $G$ is optimal 2-planar  if and only if the planar
skeleton is a 3-connected pentangulation \cite{bkr-optimal-17}, so
that $G$ is obtained by adding a pentagram of crossed edges in each
pentagon. Every optimal 2-planar graph is an optimal fan-crossing
\cite{b-fan+fcf-18}, fan-planar \cite{ku-dfang-14}, 1-gap
planar\cite{bbc-1gap-18}, and 5-map graph \cite{b-5maps-19}. There
are optimal 2-planar graphs for 
every $n \geq 50$ with $n=2$ mod 3 \cite{pt-gdfce-97}. Optimal
2-planar graphs can be recognized in linear time \cite{fkr-o2p-21}.
\item An $n$-vertex graph is optimal 1-fan-bundle  if it has $\frac{13}{3}(n-2)$ edges
\cite{abkks-fanbundle-18}. Every optimal 1-fan bundle graph consists
of a planar pentangulation, so that four crossing edges are inserted
in each pentagon (without creating a multi-edge). There are optimal
1-fan bundle graphs if $n=2$ mod 3 for properly chosen values of $n$
\cite{abkks-fanbundle-18}.
\item  An $n$-vertex graph is optimal RAC if it has $4n-10$ edges \cite{del-dgrac-11}.
Every optimal RAC graph is 1-planar \cite{el-racg1p-13}. The outer
face of a drawing is a triangle.  There are optimal RAC graphs if
$n=3k-5, k \geq 2$ \cite{del-dgrac-11}.

\end{enumerate}
\end{proposition}

\begin{table}[t]
\centering
\begin{tabular}{ l |       c |       c}
 type & density &  range of optimal graphs\\
 \hline
 1-planar & $4n-8$ \cite{bsw-1og-84} & $n=8$ and $n \geq 10$ \cite{bsw-1og-84} \\
 IC-planar & $\frac{13}{4}n-6$ \cite{zl-spgic-13} & $n=4k, k \geq 2$ \cite{zl-spgic-13}\\
 NIC-planar & $\frac{18}{5}(n-2)$ \cite{z-dcmgprc-14} & $n=5k+2, k
 \geq 2$ \cite{bbhnr-NIC-17, cs-tc1pg-14}\\
 2-planar & $5n-10$ \cite{pt-gdfce-97} & $n=20$ and $n =3k+2, k \geq
 8$   (*)\\
 1-fan-bundle & $\frac{18}{5}(n-2)$ \cite{abkks-fanbundle-18} &
  $n = 3k+2, k\geq 2$ (*)\\
RAC & $4n-10$ \cite{del-dgrac-11} &  $n \geq 4$ (*)\\

\hline
\end{tabular}
 \caption{Some types of beyond-planar graphs, their
 density and optimal graphs. A (*) indicates an extension of the
 range in this work.
 }
  \label{table}
\end{table}

\section{Combinatorial Properties} \label{sect:comb}

We now improve some results from Proposition \ref{prop1}.

\begin{lemma} \label{lem:IC}
 An IC-planar drawing of an optimal IC-planar graph consists of
  triangles and $X$-quadrangles. Every vertex is in a triangle and in exactly one
  $X$-quadrangle. 
  $X$-quadrangles are vertex disjoint. Every uncrossed edge is in a
  triangle.
  There are optimal IC-planar graphs with exponentially many
  IC-planar embeddings.
\end{lemma}

\begin{proof}
X-quadrangles are vertex disjoint by IC-planarity. An optimal
IC-planar graph has $\frac{n}{4}$ X-quadrangles, so that every
vertex is in an X-quadrangle. Also every uncrossed edge must be in
some triangle, since X-quadrangles do not share vertices or edges.
The graph    in Fig.~\ref{fig:IC-ambig}  is optimal IC-planar and
has two IC-planar embeddings. By taking $k$ copies  and a subsequent
triangulation for optimality, there is an IC-planar graph which has
  $2^k$ many embeddings, where $k=n/8$.
\qed
\end{proof}

\begin{figure}[t]
  \centering
    \includegraphics[scale=0.4]{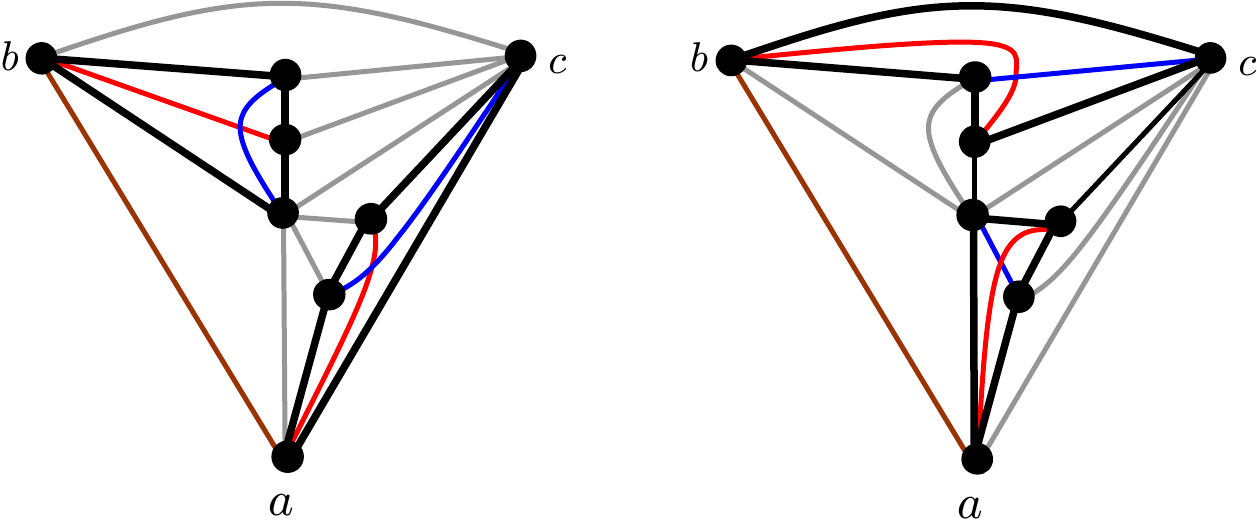}
 \caption{An IC-planar graph with two IC-planar embeddings (from
 \cite{b-IC+NIC-18}).
  }
\label{fig:IC-ambig}
\end{figure}

\begin{lemma} \label{lem:NIC}
A NIC-planar drawing of an optimal IC-planar graph consists of
  triangles and $X$-quadrangles, so that every edge is in an $X$-quadrangle
  and every uncrossed edge is in a triangle.
\end{lemma}
\begin{proof}
An optimal NIC-planar graph has $\frac{3}{5}(n-2)$ X-quadrangles and
$\frac{18}{5}(n-2)$ edges. Since two X-quadrangles do not share an
edge in a NIC-planar embedding, every edge is in an X-quadrangle and
every uncrossed edge is a triangle. Hence, there is an X-quadrangle
on one side of each uncrossed edge and a triangle on the other side.
Note that the claim can also be obtained from Corollary 4 in
\cite{bbhnr-NIC-17}.
\qed
\end{proof}

The full range for optimal 1-planar, IC-planar and NIC-planar graphs
has been discovered before, but only partially for  2-planar,
1-fan-bundle, and RAC graphs, as stated in Proposition~\ref{prop1}.

\begin{theorem} \label{lem:2-planar}
There are optimal 2-planar graphs exactly for $n=20$ and for every
$n \geq 26$ if $n=2$ mod  3.
\end{theorem}
\begin{proof}
There are 5-connected 5-regular planar graphs with $n$ faces for
$n=20$ and for every $n \geq 26$ if $n=2$ mod  3, as shown by
Hasheminezhad et al.~\cite{hmr-pentagons-11}. The dual is a
3-connected  pentangulation, which is turned into an optimal
2-planar if each pentagon is filled by a pentagram, as shown in
\cite{bkr-optimal-17}. The smallest pentangulations are shown in
Fig.~\ref{fig:pentangulation}.
\qed
\end{proof}

\begin{figure}[t]
  \centering
  \subfigure[]{
    \includegraphics[scale=0.7]{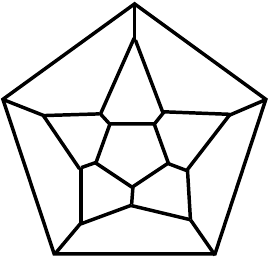}
    \label{fig:dodecahedon}
    }
    \hspace{1mm}
 \subfigure[]{
        \includegraphics[scale=0.7]{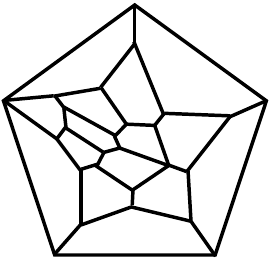}
\label{fig:penta26}
     }
        \hspace{1mm}
 \subfigure[]{
       \includegraphics[scale=0.7]{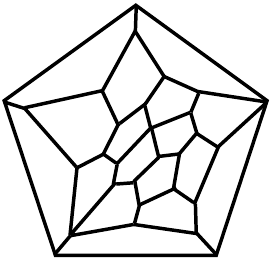}
\label{fig:penta29} }
\hspace{1mm}
 \subfigure[]{
        \includegraphics[scale=0.7]{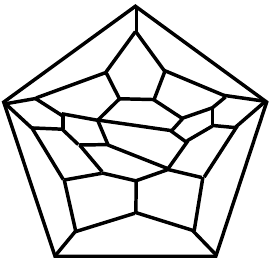}
\label{fig:penta32}
     }
 \caption{The smallest pentangulations with 20, 26, 29 and 32 vertices.
  }
  \label{fig:pentangulation}
\end{figure}

Note that there is no optimal 1-planar graph with seven or nine
vertices  and no optimal 2-planar graph with 23
vertices.\\

\begin{theorem} \label{lem:2-planar}
There are optimal 1-fan-bundle graphs exactly  for every $n \geq 8$
if $n=2$ mod  3.
\end{theorem}
\begin{proof}
Since optimal 1-fan-bundle graphs have $\frac{13}{3}(n-2)$ edges,
there are optimal graphs only for $n=3k+2$ by integrality.   The
5-clique minus one edge has only nine edges, so that it is not
optimal 1-fan-bundle.  Optimal graphs for $n = 3k+2$ and $k \geq 2$
can be obtained from a planar pentangulation and the insertion of
four edges in each pentagon. A pentangulation may have vertices of
degree two, as opposed to the previous case for 2-planar graphs.
Then the neighbors of  a degree two vertex are connected by an edge
in only one of the adjacent faces. The smallest  optimal
1-fan-bundle graph is shown in Fig.~\ref{fig:fanbundle}. By
induction, remove the crossed edges in a pentagon $P$, add three
vertices as Steiner points, partition $P$ into three pentagons, so
that there are two vertices of degree two, and insert four crossed
edges in each face, so that there is no multi-edge, see
 Fig.~\ref{fig:fanbundle}. So we obtain
optimal graphs for every $n=3k+2, k\geq 2$.
\qed
\end{proof}

A \emph{primal-dual graph} of a  3-connected planar graph $G$ is
obtained by the simultaneous drawing of $G$ and its dual $G^*$, from
which the dual vertex for the outer face has been removed.  Every
primal edge of $G$ is crossed by the dual edge between the faces on
either side. Moreover, every dual vertex for a face of $G$ is
adjacent to every primal vertex in the boundary of the face, see
Fig.~\ref{fig:opt-RAC-9}. A primal-dual graph is 1-planar.
Primal-dual graphs (including the outer face) have been used by
Ringel \cite{ringel-65} in his early study of 1-planar graphs.
Brightwell and Scheinerman \cite{bs-rpg-93} have shown that
primal-dual  graphs admit a RAC drawing, see also
\cite{dlq-circlepacking-20, fr-pdcr-19}.

\begin{figure}[t]
  \centering
  \subfigure[]{
    \includegraphics[scale=0.7]{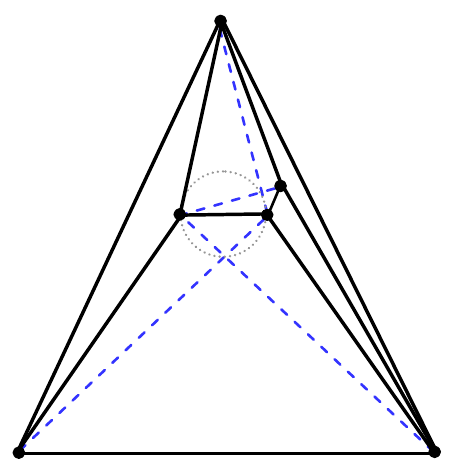}
    \label{fig:optRAC-6}
    }
    \hspace{2mm}
 \subfigure[]{
        \includegraphics[scale=0.65]{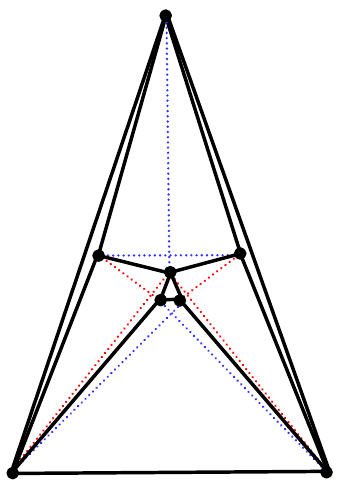}
\label{fig:optRAC-8}
     }
        \hspace{2mm}
 \subfigure[]{
        \includegraphics[scale=0.7]{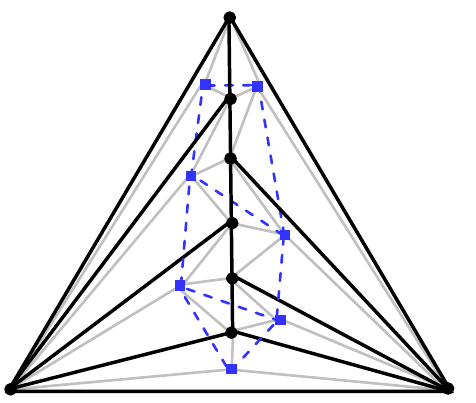}
\label{fig:opt-RAC-9} }
 \caption{RAC drawings of optimal RAC graphs with (a) six and (b)
 eight vertices, and (c) a vertex-face graph. 
  }
  \label{fig:optRAC}
\end{figure}

\begin{theorem} \label{lem:RAC-there are}
There are optimal RAC graphs  for every $n \geq 4$.
\end{theorem}

\begin{proof}
Clearly, $K_4$ and $K_5$ are optimal RAC graphs. Optimal RAC graphs
with six and eight vertices are displayed in Fig.~\ref{fig:optRAC}.
In general, an optimal RAC graph  can be constructed from a
primal-dual graph, as observed by Didimo et al~\cite{del-dgrac-11}.
In particular, optimal RAC graphs with $n \geq 9$ can be constructed
a follows. Add $k \geq 2$ vertices $v_1,\ldots, v_k$ in the interior
of an outer triangle $\Delta(a,b,c)$ with $c$ on top. Add edges
$\edge{c}{v_k}$ and $\edge{v_i}{v_{i+1}}$ for $i=1,\ldots, k-1$, so
that there is a path. Next add edges so that the so obtained graph
is 3-connected and inner faces are triangles or quadrangles. For
example, add edges $\edge{a}{v_1}, \edge{b}{v_1}$ and
$\edge{a}{v_{2i+1}}$ and $\edge{b}{v_{2i}}$ for $1 \leq i \lfloor
(k+1)/2 \rfloor$, so that there are $k-1$ quadrangles in an
$n$-vertex graph with $n=k+3$.
 Then there are three inner triangles and $n-4$
quadrangles. The primal-dual graph has $2n-1$ vertices and
 $(n-1)+(n-4)+4(n-3)+9 = 4n-12$ edges, and it is 1-planar
\cite{ringel-65}. One more vertex and four more edges are obtained
by the insertion of a diagonal in a quadrangle. The so obtained
primal-dual graph is 3-connected and admits a RAC drawing
\cite{bs-rpg-93}, so that there is an optimal RAC graph for every $n
\geq 9$. The one for $n=7$ is obtained in the same way, see
\cite{del-dgrac-11}.
\qed
\end{proof}

\section{Minors} \label{sect:minor}

It is well-known that every graph has a subdivision that is
1-planar. In other words, every graph is a topological minor of a
1-planar graph. This fact has been improved to IC-planar graphs and
to upper bounds on the number of subdivisions  \cite{b-fcf-21}. In
particular,  every graph has a 3-subdivision that is   RAC
\cite{del-dgrac-11},  a 2-subdivision that is fan-crossing, and a
1-subdivision that is quasi-planar \cite{b-fcf-21}. We consider
minors in optimal graphs.

\begin{theorem} \label{thm:optimal}
For every graph  $G$ there is an optimal beyond-planar graph $H$ of
type $\tau$, where $\tau$ is IC-planar, NIC-planar, 1-planar,
2-planar, 1-fan-bundle and RAC, respectively, such that $G$ is a
topological minor of $H$.

\end{theorem}

\begin{proof}
Consider a  drawing of $G$ and treat it as a planar graph, so that
every crossing point of two edges of $G$ is a new vertex of degree
four. We construct a host graph $H$ by placing a ``gadget'' at every
crossing point, so that the crossing happens in the gadget.
Thereafter,  the intermediate graph is augmented for optimality. The
gadget is an X-quadrangle if the type is 1-planar, IC-planar,
NIC-planar, and RAC, respectively, and an X-pentagon or the
dodecahedron graph with crossing edges in each inner face for
2-planar and 1-fan-bundle graphs.

First, for IC-planar graphs, we create an X-quadrangle for every
vertex $v$ of $G$ so that $v$ is one vertex of the X-quadrangle,
and, as aforesaid, an X-quadrangle at every crossing. Every edge of
$G$ is partitioned into segments, which are uncrossed pieces in the
drawing. These segments are inherited by $H$, so that a segment
connects  vertices in two X-quadrangles. There are no further
vertices in $H$, so that every vertex is in an X-quadrangle. Two
 X-quadrangles are vertex disjoint, which guaranties IC-planarity.
 By Lemma~\ref{lem:IC}, optimality is obtained by a triangulation.
  Clearly, for every edge of $G$ there is a path in $H$,
so that two such paths  are vertex disjoint.

Similarly, for  NIC-planar graphs, we replace each crossing  by an
X-quadrangle, so that the X-quadrangles for two consecutive
crossings along an edge share a vertex. In other words, segments are
contracted. The first (last) segment of an edge incident to a vertex
of $G$ is replaced by an X-configuration with the segment as a
diagonal. Thereby, every edge of $G$ is subdivided, so that two
paths for edges of $G$ are vertex disjoint in $H$. It remains to
construct an optimal NIC-planar graph by filling the remaining
faces. In addition, we wish to keep the degree low. So far, the
boundary of each face consists of edges from X-quadrangles. If there
is a triangle, we are done. The gadget from Fig.~\ref{fig:NIC} is
inserted in a quadrangle. Larger faces are partitioned using
X-quadrangles and triangles and the gadget for quadrangles, so that
each edge of $H$ is in an X-quadrangle and in a triangle. Then $H$
is an optimal NIC-planar graph by Lemma~\ref{lem:NIC}.

For 1-planarity,  the size of the  faces of the drawing of $G$ is
even  if the crossing points are replaced by X-quadrangles and the
first (last) segment between a vertex and a crossing point is
subdivided. Then the faces  can be partitioned into quadrangles,
which are  augmented to X-quadrangles, see \cite{s-K7minors-17}. The
planar skeleton is 3-connected \cite{abk-sld3c-13}, so that $H$ is
optimal 1-planar.

For 2-planarity, the drawing of $G$ must   be transformed into a
3-connected pentangulation. Therefore, we first replace every
crossing point  by the dodecahedron graph, so that the segment of an
edge between two crossing points is attached to two outer vertices
on opposite sides of the outer pentagon, or two vertices of adjacent
pentagons coincide, as before in the NIC-planar case. If every inner
pentagram is filled by a pentagram, then two crossing edges can be
routed internally, so that their
subdivisions are vertex disjoint. 
Hence, there is a subdivision of $G$. Towards optimality, large
faces of size at least six are partitioned into pentagrams,
triangles and quadrangles, so that there are no vertices of degree
two and the intermediate graph is 3-connected. If there is a
quadrangle, then insert the dodecahedron graph in its interior and
partition the region in between into
three more pentagons. 
Similarly, the region between a triangle and an inserted
dodecahedron graph is partitioned into two pentagons and a
quadrangle, which is partitioned into pentagons as described before.
Finally, all pentagons are filled by pentagrams, so that there are
only X-pentagons. This does not create a multi-edge, since every
vertex of the planar skeleton  has degree at least three, and there
is no separation pair, so that the pentangulation  is 3-connected.
Hence, there is an optimal 2-planar graph that contains a
subdivision of $G$ as a subgraph.

For 1-fan-bundle graphs, we proceed as before, but finally fill the
pentagrams by four edges, so that an edge crossing in $G$ is
transferred to an edge crossing inside a pentagon.

For RAC graphs,  suppose that $G$ is drawn with  straight-line
segments. Treat the drawing as a planar graph and triangulate it.
Then subdivide each edge, place a new vertex in each face (except
for the outer face) and connect the inserted vertex with the six
vertices in the boundary of the face, so that there is a
re-triangulation. Now construct the primal-dual graph $H$, which is
  1-planar graph and an optimal RAC graph, since it admits a RAC
  drawing as shown by
Brightwell and Scheinerman~\cite{bs-rpg-93}. It remains to consider
paths for edges of $G$. Every edge of $G$ has a subdivision in $H$
which uses the segments from the drawing. Then two paths meet in the
crossing point for their  edges. This collision is circumvented
using a bypass. Suppose that edges $e, f$ and $g$ cross each other
so that there is a triangle in the drawing of $G$ with crossing
points $u, v$ and $w$,   see Fig.~\ref{fig:bypass}. All other cases
are similar. Then $e$ can bypass $u$ and pass through $w$, $f$ can
pass through $u$ and bypass $v$, and $h$ can bypass $w$ and pass
through $v$. For edge $e$, the bypath begins at the subdivision
point just before $u$, it goes through the (vertices for the) faces
next to $u$, and ends at the subdivision point on $e$ after $u$.
Thereby it crosses the path for edge $g$. Hence, there are vertex
disjoint paths in the vertex face graph $H$ for the edges of $G$, so
that a subdivision of $G$ is a subgraph of $H$, that is $G$ is a
topological minor of $H$.
\qed
\end{proof}

\begin{figure}[t]
  \centering
\subfigure[]{
        \includegraphics[scale=0.7]{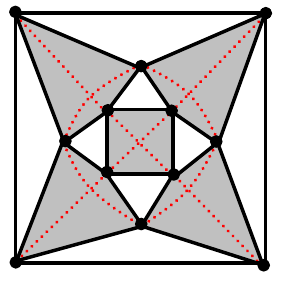}
\label{fig:NIC}
     }
     \hspace{1mm}
 \subfigure[]{
        \includegraphics[scale=0.7]{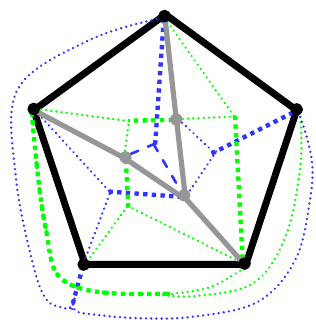}
\label{fig:fanbundle}
     }
    \hspace{1mm}
  \subfigure[]{
        \includegraphics[scale=0.7]{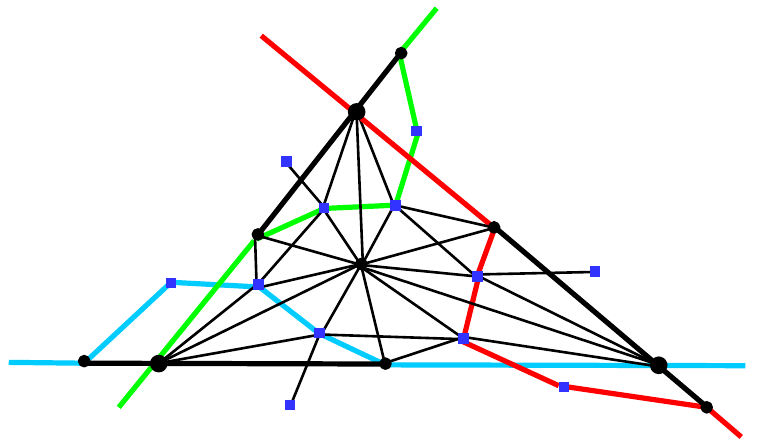}
\label{fig:bypass}
      }
 \caption{Illustration to the proof of Theorem~\ref{thm:optimal}.
(a)   A NIC-planar graph in a quadrangle. (b) The
 smallest optimal 1-fan-bundle graph with three Steiner points (gray
 in the electronic version) for a pentangulation of a pentagon.
(c) A  bypass for three pairwise
 crossing edges using a primal-dual graph and auxiliary vertices.
 In an electronic version, the edges (paths) are colored red, blue and green.
  }
   \label{fig:thm1}
\end{figure}

\begin{corollary} \label{cor:optimal-minor}
Every graph  is a topological minor of an optimal graph for each of
the following types: fan-crossing free, grid-crossing, fan-crossing,
fan-planar, 1-gap-planar, 4-map, 5-map.
\end{corollary}
\begin{proof}
We can use the constructions for 1-planar and 2-planar graphs from
(the proof of) Theorem~\ref{thm:topological-minor}, since optimal
fan-crossing free (grid crossing) graphs are optimal 1-planar
\cite{cpkk-fan-15}, and every 4-map graph is a 3-connected
triangulated 1-planar graph. Similarly, the crossed dodecahedron
graph is   a 5-map graph \cite{b-5maps-19} and simultaneously
2-planar and fan-crossing, so that the construction for 2-planar
graphs can be used for all types of
 graphs with a density of $5n-10$.
 \qed
\end{proof}

Finally, we   distinguish topological minors from  minors in optimal
beyond-planar graphs. Similar to the case of $K_5$ and graphs of
degree at most three, we wish to keep the degree of graphs low if
they contain a (large) clique as a minor. The degree is determined
by the gadgets for edge crossings and the filling of faces towards
optimality. The degree of the minor can be decreased to three.

A 3-\emph{regularization} transforms a graph  into a graph of degree
three by a local operation on vertices. Examples are the
\emph{node-to-circle expansion}, which expands every vertex   of
degree $d$ into a circle of $d$ vertices of degree three
\cite{b-fcf-21}, or the expansion into a binary tree with $d$
leaves. Each vertex on the circle (leaf) inherits one incident edge.
3-regularization preserves minors, so that $G$ is a minor of a
3-regularization $\eta(H)$  if $G$ is a minor of $H$. However,
3-regularization  does not preserve topological minors, since the
obtained graphs have degree at most three, and therefore exclude any
graph with a vertex of degree at least four as a topological minor.
In consequence, $K_5$ is not a topological minor of any
3-regularization.

\begin{theorem} \label{thm:topological-minor}
For every type $\tau$ of beyond-planar graphs as above,  there is a
constant $d_{\tau}$ and an optimal $\tau$-graph $H$ so that the
complete graph $K_k$ with $k=d_{\tau}$ is a minor but not a
topological minor of $H$. In the above cases, we have $d_{\tau} \leq
42$.
\end{theorem}
\begin{proof}
For every type $\tau$, there is an optimal  $\tau$-graph $H$  so
that the 3-regularization of   $K_{d}$ is a (topological) minor of
$H$ by Theorem~\ref{thm:optimal}.
Recall the construction of $H$ from the proof of
Theorem~\ref{thm:optimal} and try to keep the degree low.

For IC-planar graphs, there is a triangulation of faces that
increases the degree of each vertex by at most   two, and there are
at most two faces for each vertex, since there are gadgets for the
crossing points. Hence, $H$ has degree at  most   eight,  so that
there are optimal IC-planar graphs that contain $\eta(K_{10})$ as a
minor, but not as a topological minor.

Similarly, the degree of $H$ can be kept as low as 24 for NIC-planar
graphs, and  14 for 1-planar graphs, since X-quadrangles and gadgets
are used for a partition of large faces.

For RAC graphs, we first subdivide each edge of the planar graph
from the drawing and create triangles with two subdivision points
and a vertex or a crossing point from the drawing, so that each
crossing point is surrounded by four triangles. Then  we triangulate
the remaining faces, so that the degree of each subdivision point is
increased at most by two from the face on either side.  The
intermediate graph is a triangulated planar graph of  degree at most
ten. In the next step we use a primal-dual graph for each triangle,
so that   graph $H$ has degree at most 40.

For 2-planar and 1-fan-bundle graphs, there is a pentagon around
each vertex and each crossing point in  the drawing of $G$. Faces
can be  filled, so that at most most four dodecahedron graphs meet
in a point. Thereby, graph $H$ has degree at most 12.

In any case, graph $H$ does contain a topological minor of degree at
least $d+1$, so that $K_{d+2}$ is a  minor but not a topological
minor of $H$.
\qed
\end{proof}

\section{Conclusion} \label{sect:conclusion}

In this work, we  study optimal graphs for some important types of
beyond-planar graphs. We   compute  the range for such graphs and
 show  that optimal graphs contain any graph as a (topological)
minor.

Our results can be extended to 3-planar and 1-bend RAC graphs if
such optimal graphs have $5.5n-15$ edges and consist of hexagons
with all but one diagonal from $K_6$ in the interior. However, new
techniques are needed for graphs with a density beyond $6n$, since
there is no longer a partition into a planar skeleton and edges
crossing in the interior of small faces.

Open problems include the characterization of optimal fan-crossing
and optimal 1-gap planar graphs, and  the recognition problem for
optimal RAC graphs.

\section{Acknowledgements} \label{sect:ack}
I wish to thank David Eppstein for drawing my attention to circle
packing and meshes, Gunnar Brinkmann for a reference to
\cite{hmr-pentagons-11} and providing the pentangulation with 26
vertices, and Brendan KcKay for
providing the (dual of the) pentangulation with 29 vertices.\\


\bibliographystyle{abbrv}
\bibliography{brandybibV8a}

\end{document}